\documentclass[12pt]{article}

\usepackage[margin=1in]{geometry}
\usepackage{CJKutf8, amsmath, amssymb, amsthm, authblk, complexity, hyperref}

\newtheorem{lemma}{Lemma}
\newtheorem{theorem}{Theorem}

\theoremstyle{definition}
\newtheorem{assumption}{Assumption}
\newtheorem{definition}{Definition}
\newtheorem{example}{Example}

\theoremstyle{remark}
\newtheorem*{remark}{Remark}

\DeclareMathOperator{\tr}{tr}
\DeclareMathOperator{\supp}{supp}

\usepackage[style=trad-unsrt, citestyle=numeric-comp, giveninits=true, doi=false, url=false, eprint=false, isbn=false]{biblatex}
\begin{document}

\title{Convergence of eigenstate expectation values with system size}

\begin{CJK}{UTF8}{}

\CJKfamily{gbsn}

\author{Yichen Huang (黄溢辰)\thanks{yichuang@mit.edu}}
\affil{Center for Theoretical Physics, Massachusetts Institute of Technology, Cambridge, Massachusetts 02139, USA}

\maketitle

\end{CJK}

\begin{abstract}

Understanding the asymptotic behavior of physical quantities in the thermodynamic limit is a fundamental problem in statistical mechanics. In this paper, we study how fast the eigenstate expectation values of a local operator converge to a smooth function of energy density as the system size diverges. In translation-invariant quantum lattice systems in any spatial dimension, we prove that for all but a measure zero set of local operators, the deviations of finite-size eigenstate expectation values from the aforementioned smooth function are lower bounded by $1/O(N)$, where $N$ is the system size. The lower bound holds regardless of the integrability or chaoticity of the model, and is saturated in systems satisfying the eigenstate thermalization hypothesis.

\end{abstract}

\section{Introduction}

Many predictions of statistical mechanics require taking the thermodynamic limit, and such results are usually exact or universal only in this limit. Therefore, it is important and fundamental to understand how physical quantities approach their values in the thermodynamic limit as the system size diverges.

In this paper, we study the eigenstate expectation values (EEV) of local operators in translation-invariant (TI) quantum lattice systems. TI allows us to define an infinite sequence of Hamiltonians, one for each system size, from a fixed local term (in the Hamiltonian). However, the thermodynamic limit of EEV is not yet well defined. Since the number of eigenstates grows exponentially with the system size, it is not immediately clear how to naturally define a sequence of eigenstates (one for each system size) in which the convergence of local expectation values is to be studied.

The eigenstate thermalization hypothesis (ETH) \cite{Deu91, Sre94, RDO08, DKPR16, Deu18} postulates that in the thermodynamic limit, the EEV of a local operator converges to a smooth function of energy density. If this is true and such a function is known, we can compute the deviations of finite-size EEV from the values of the function at the same energy density, and analyze how the deviations vanish as the system size diverges. If such a function is unknown (or even the ETH is false), we need to search the space of smooth functions and find the optimal ``target function'' such that the deviations of finite-size EEV from the target function decay as fast as possible in the thermodynamic limit.

In TI systems in any spatial dimension, we prove that for all but a measure zero set of local operators, the deviations of finite-size EEV from the best target function (which depends on the local operator under consideration) are lower bounded by $1/O(N)$, where $N$ is the system size. Note that this result does not assume the ETH. If in the thermodynamic limit the EEV of a local operator does not converge to a smooth function of energy density, then the deviations from any target function do not vanish and the lower bound is trivially valid. In systems satisfying the ETH, we prove that the bound is saturated.

The rest of this paper is organized as follows. Section \ref{sec:pre} sets the stage. In Section \ref{sec:def}, we rigorously define the convergence rate of EEV in the thermodynamic limit. Section \ref{sec:res} presents the main results, whose relationship with the ETH is discussed in Section \ref{sec:eth}. In particular, we explain why our results do not contradict the conventional wisdom that the fluctuations of EEV in chaotic systems are exponentially small in the system size \cite{Sre99, KIH14}. Section \ref{sec:con} concludes the paper. The main text of this paper should be easy to read, for most of the technical details are deferred to Appendices \ref{app}, \ref{app:B}.

\section{Preliminaries} \label{sec:pre}

Throughout this paper, standard asymptotic notations are used extensively. Let $f,g:\mathbb R^+\to\mathbb R^+$ be two functions. One writes $f(x)=O(g(x))$ if and only if there exist constants $M,x_0>0$ such that $f(x)\le Mg(x)$ for all $x>x_0$; $f(x)=\Omega(g(x))$ if and only if there exist constants $M,x_0>0$ such that $f(x)\ge Mg(x)$ for all $x>x_0$; $f(x)=\Theta(g(x))$ if and only if there exist constants $M_1,M_2,x_0>0$ such that $M_1g(x)\le f(x)\le M_2g(x)$ for all $x>x_0$; $f(x)=o(g(x))$ if and only if for any constant $M>0$ there exists a constant $x_0>0$ such that $f(x)<Mg(x)$ for all $x>x_0$; $f(x)=\omega(g(x))$ if and only if for any constant $M>0$ there exists a constant $x_0>0$ such that $f(x)>Mg(x)$ for all $x>x_0$.

For notational simplicity and without loss of generality, we present our results in one dimension. (It is easy to see that the same results hold in higher spatial dimensions.) Consider a chain of $N$ spins so that the dimension of the Hilbert space is $d=d_\textnormal{loc}^N$, where $d_\textnormal{loc}=\Theta(1)$ is the local dimension of each spin. The system is governed by a TI $k$-local Hamiltonian $H$. TI implies periodic boundary conditions, and ``$k$-local'' means that the support of each term in $H$ is contained in a consecutive region of size $k=\Theta(1)$. (For example, a term acting nontrivially only on the first and third spins is $3$- rather than $2$-local.) We say that a term is exactly $k$-local if and only if it is $k$-local but not $(k-1)$-local. Let $\mathbb T$ be the (unitary) lattice translation operator, which acts on the computational basis states as
\begin{equation}
\mathbb T(|x_1\rangle\otimes|x_2\rangle\otimes\cdots\otimes|x_N\rangle)=|x_2\rangle\otimes|x_3\rangle\otimes\cdots\otimes|x_N\rangle\otimes|x_1\rangle
\end{equation}
with $x_l\in\{0,1,\ldots,d_\textnormal{loc}-1\}$ for $l=1,2,\ldots,N$. We write the Hamiltonian as
\begin{equation}
H=\sum_{l=0}^{N-1}H_l,\quad H_l=\mathbb T^{-l}h\mathbb T^l,
\end{equation}
where $h$ is a Hermitian operator acting on the first $k$ spins. Assume without loss of generality that $\tr h=0$ (traceless) and $\|h\|=1$ (unit operator norm).

\begin{lemma} \label{l:1}
For any traceless $k'$-local operator $A$, both $\tr(HA)/d$ and $\tr(H^2A)/d$ are $N$-independent constants for $N\ge k+k'-1$ and $N\ge2k+k'-2$, respectively. Furthermore, $\tr(Hh)/d$ is an $N$-independent positive constant for $N\ge2k-1$.
\end{lemma}

\begin{proof}
Let $\supp\cdots$ be the support of a local operator. Since $\tr H_l=0$ for all $l$ and $\tr A=0$,
\begin{gather}
\frac1d\tr(HA)=\frac1d\sum_{\supp H_l\cap\supp A\neq\emptyset}\tr(H_lA),\label{eq:1st}\\
\frac1d\tr(H^2A)=\frac1d\sum_{\substack{((\supp H_{l_1}\cap\supp A\neq\emptyset)\land(\supp H_{l_2}\cap(\supp H_{l_1}\cup\supp A)\neq\emptyset))\\\lor((\supp H_{l_2}\cap\supp A\neq\emptyset)\land(\supp H_{l_1}\cap(\supp H_{l_2}\cup\supp A)\neq\emptyset))}}\tr(H_{l_1}H_{l_2}A).\label{eq:2nd}
\end{gather}
It is easy to see that the right-hand sides of Eqs. (\ref{eq:1st}), (\ref{eq:2nd}) do not depend on $N$ for $N\ge k+k'-1$ and $N\ge2k+k'-2$, respectively. Due to TI, $\tr(Hh)/d=\tr(H^2)/(Nd)>0$.
\end{proof}

Since we are interested in the thermodynamic limit, hereafter we only consider sufficiently large $N$ such that conditions like $N\ge2k+k'-2$ are satisfied.

\section{Definitions} \label{sec:def}

Let $\{|j\rangle\}_{j=1}^d$ be a complete set of TI eigenstates of $H$ with corresponding energies $\{E_j\}$. Note that both $|j\rangle$ and $E_j$ depend on the system size and should carry $N$ as a subscript, which is omitted for notational simplicity.

\begin{definition} [Convergence rate of eigenstate expectation values] \label{def}
For a traceless local operator $A$ with $\|A\|=1$, let $f:[-1,1]\to\{z\in\mathbb C:|z|\le1\}$ be an $N$-independent function and define
\begin{equation} \label{asmpeq}
r_f(N)=\sqrt{\frac1d\sum_j\big|\langle j|A|j\rangle-f(E_j/N)\big|^2},\quad R_f(N)=\sup_{n\ge N}r_f(n)\ge r_f(N),
\end{equation}
where $f$ is smooth in the sense of having a Taylor expansion to second order around $x=0$:
\begin{equation}
f(x)=f(0)+f'(0)x+f''(0)x^2/2+O(x^3).
\end{equation}
If there exists an optimal $\hat f$ such that $R_{\hat f}(N)=O(R_f(N))$ for any (other) smooth function $f$, then the decay of $R_{\hat f}(N)$ gives the (average) convergence rate of the EEV of $A$ in the thermodynamic limit $N\to+\infty$.
\end{definition}

Note that $r_f(N)$, $R_f(N)$, and $\hat f$ (if exists) all depend on the local operator under consideration and should carry $A$ as a subscript, which is omitted for notational simplicity.

\begin{example} \label{exa}
In the special case where $A=h$ is a term in the Hamiltonian, we trivially have $\hat f(x)=x$ and $R_{\hat f}(N)=0$ for any $N$.
\end{example}

For any traceless local operator $A$ in any TI system, the weak ETH \cite{BKL10, Mor16, BCSB19} (Lemma \ref{weaketh}) implies that $R_{f(x)=0}(N)=O(1/\sqrt N)$. If an optimal $\hat f$ exists, then we obtain an upper bound $R_{\hat f}(N)=O(1/\sqrt N)$.

\begin{lemma} [\cite{KLW15, HH19}] \label{weaketh}
For any traceless local operator $A$ with bounded norm $\|A\|=O(1)$,
\begin{equation}
\frac1d\sum_j\big|\langle j|A|j\rangle\big|^2=O(1/N).
\end{equation}
\end{lemma}

\begin{proof}
We include a proof for completeness. Let
\begin{equation} \label{TI}
\mathbb A:=\frac1N\sum_{l=0}^{N-1}\mathbb T^{-l}A\mathbb T^l
\end{equation}
so that $\langle j|A|j\rangle=\langle j|\mathbb A|j\rangle$ due to TI. Hence,
\begin{equation}
\sum_j\big|\langle j|A|j\rangle\big|^2=\sum_j\langle j|\mathbb A^\dag|j\rangle\langle j|\mathbb A|j\rangle\le\sum_{j,k}\langle j|\mathbb A^\dag|k\rangle\langle k|\mathbb A|j\rangle=\sum_j\langle j|\mathbb A^\dag\mathbb A|j\rangle=\tr(\mathbb A^\dag\mathbb A).
\end{equation}
Expanding $\mathbb A$ in the generalized Pauli basis, we count the number of terms that do not vanish upon taking the trace in the expansion of $\mathbb A^\dag\mathbb A$. There are $O(N)$ such terms, the trace of each of which is $O(d/N^2)$. Therefore, $\tr(\mathbb A^\dag\mathbb A)/d=O(1/N)$.
\end{proof}

\section{Results} \label{sec:res}

We prove the following lemma in Appendix \ref{app}.

\begin{lemma} \label{thm}
For a traceless local operator $A$ with $\|A\|=1$, if there exist a smooth function $f$ and a strictly increasing infinite sequence $\{N_i\}$ of positive integers such that $r_f(N)=o(1/N)$ for $N\in\{N_i\}$, then
\begin{equation} \label{wrong}
\tr(Hh)\tr(H^2A)/d^2=\tr(H^2h)\tr(HA)/d^2.
\end{equation}
Note that both sides of this equation are $N$-independent constants (Lemma \ref{l:1}).
\end{lemma}

We need to define a measure on the set of local operators or parameterize a local operator by real numbers. TI allows us to define canonical local operators, which not only are representatives of all local operators but also form a vector space. Expanding a traceless local operator $A$ in the generalized Pauli basis, we say that $A$ is canonical if and only if all Pauli string operators (with non-zero coefficients) in the expansion start from the first site. For any $A$, there is a unique canonical traceless local operator $B$, called the canonical form of $A$, such that $\langle\psi|A|\psi\rangle=\langle\psi|B|\psi\rangle$ for any TI state $|\psi\rangle$. For example, in a spin-$1/2$ chain $\sigma_2^z\sigma_3^z+\sigma_2^x+\sigma_3^x$ is not canonical, and its canonical form is $\sigma_1^z\sigma_2^z+2\sigma_1^x$, where $\sigma_l^x, \sigma_l^z$ are the Pauli matrices at site $l$. With TI, we may without loss of generality only consider the EEV of canonical traceless local operators.

The expansion of a general canonical traceless $k'$-local operator in the generalized Pauli basis has $(d_\textnormal{loc}^2-1)$ exactly $1$-local terms and $(d_\textnormal{loc}^2-1)^2d_\textnormal{loc}^{2\kappa-4}$ exactly $\kappa$-local terms for $\kappa=2,3,\ldots,k'$. The coefficients of the expansion parameterize a canonical traceless $k'$-local operator. Thus, we have defined a parameter space $S$ of dimension
\begin{equation}
d_\textnormal{loc}^2-1+\sum_{\kappa=2}^{k'}(d_\textnormal{loc}^2-1)^2d_\textnormal{loc}^{2\kappa-4}=(d_\textnormal{loc}^2-1)d_\textnormal{loc}^{2k'-2}
\end{equation}
such that points in $S$ are in one-to-one correspondence with canonical traceless $k'$-local operators.

\begin{theorem} \label{cor}
We say that a canonical traceless local operator $A$ is ``rapidly converging'' if there exists a smooth function $f$ (which depends on $A$) such that
\begin{equation} \label{eq:thm}
r_f(N)=1/O(N)
\end{equation}
does not hold. The set of rapidly converging canonical traceless local operators has measure zero.
\end{theorem}

\begin{proof}
It suffices to prove that the set of canonical traceless local operators that satisfy Eq. (\ref{wrong}) has measure zero. Since both sides of Eq. (\ref{wrong}) are linear functions of $A$, it further suffices to find a particular $A$ such that Eq. (\ref{wrong}) does not hold.

Assume without loss of generality that $h$ is canonical and exactly $k$-local. We write $H^2=G_1+G_2+G_3$, where
\begin{equation}
G_1:=2\sum_{l=0}^{N-1}H_lH_{l+2k-1},\quad G_2:=\sum_{l=0}^{N-1}\sum_{\Delta=2-2k}^{2k-2}H_lH_{l+\Delta},\quad G_3:=\sum_{l=0}^{N-1}\sum_{\Delta=2k}^{N-2k}H_lH_{l+\Delta}.
\end{equation}
Expanding $G_1$ in the generalized Pauli basis, there is an exactly $(3k-1)$-local term (coming from $H_0H_{2k-1}$) whose support contains spins at positions $1,k,2k,3k-1$. Define $A$ as this term so that $\tr(G_1A)\neq0$. Moreover, $\tr(G_2A)=0$ because all terms in $G_2$ are $(3k-2)$-local. The support of each term in $G_3$ has a gap of $\ge k$ spins. The support of $A$ does not have such a gap. Hence, $\tr(G_3A)=0$, and
\begin{equation}
\tr(H^2A)=\tr(G_1A)+\tr(G_2A)+\tr(G_3A)\neq0.
\end{equation}
Since $\tr(Hh)>0$ (Lemma \ref{l:1}), the left-hand side of Eq. (\ref{wrong}) is non-zero. We complete the proof by noting that $\tr(HA)=0$ because all terms in $H$ are $k$-local.
\end{proof}

\begin{remark}
If, instead of Eq. (\ref{asmpeq}), $r_f(N)$ is defined as
\begin{equation}
r_f(N)=\frac1d\sum_j\big|\langle j|A|j\rangle-f(E_j/N)\big|,
\end{equation}
then the statement of Theorem \ref{cor} remains valid upon changing Eq. (\ref{eq:thm}) to $r_f(N)=1/O(N\log N)$. This can be proved in almost the same way. The difference comes from the observation that ``$O(Nr_f(N))$'' in Eq. (\ref{2derivc}) should be modified to $\Lambda^2r_f(N)$.
\end{remark}

\section{Eigenstate thermalization} \label{sec:eth}

The (strong) ETH postulates that the diagonal matrix elements of a local operator $A$ in the energy eigenbasis take the form \cite{Sre99}
\begin{equation} \label{eq:ETH}
\langle j|A|j\rangle=g(E_j/N)+e^{-S(E_j)/2}\delta_j,
\end{equation}
where $g(\cdots)$ is a smooth function of its argument, $S(E)$ is the thermodynamic entropy (logarithm of the density of states) at energy $E$, and $\delta_j=O(1)$ varies erratically with $j$.

Since the thermodynamic entropy is extensive, by comparing Eqs. (\ref{asmpeq}), (\ref{eq:ETH}) one might argue that $R_g(N)=e^{-\Theta(N)}$, which contradicts Theorem \ref{cor}. However, this argument is problematic because $g$ may depend on $N$. Indeed, $g$ should carry $N$ as a subscript, and Theorem \ref{cor} states that for generic local operators, $g_N$ cannot converge too fast in the thermodynamic limit $N\to+\infty$.

Interestingly, the ETH for eigenstates in the middle of the energy spectrum implies that the bound (\ref{eq:thm}) is tight.

\begin{assumption} [eigenstate thermalization hypothesis in the middle of the spectrum] \label{ass}
Let $\epsilon$ be an arbitrarily small positive constant. For any traceless local operator $A$ with $\|A\|=1$, there is a sequence of functions $\{g_N:[-\epsilon,\epsilon]\to\{z\in\mathbb C:|z|\le1\}\}$ (one for each system size $N$) such that
\begin{equation} \label{eq:ass}
\big|\langle j|A|j\rangle-g_N(E_j/N)\big|\le1/\poly(N)
\end{equation}
for all $j$ with $|E_j|\le N\epsilon$, where $\poly(N)$ denotes a polynomial of sufficiently high degree in $N$. We assume that each $g_N(x)$ is smooth in the sense of having a Taylor expansion to some low order around $x=0$.
\end{assumption}
While the ETH ansatz (\ref{eq:ETH}) implies that the right-hand side of inequality (\ref{eq:ass}) can be improved to $e^{-\Omega(n)}$, a (much weaker) inverse polynomial upper bound suffices for our purposes.

\begin{theorem} \label{thm:eth}
For a traceless local operator $A$ with $\|A\|=1$, let
\begin{equation}
f(x):=\tr(HA)x/\tr(Hh).
\end{equation}
Assumption \ref{ass} implies that
\begin{equation}
R_f(N)=O(1/N).
\end{equation}
\end{theorem}

\section{Conclusion} \label{sec:con}

In summary, we have proposed a definition of the convergence rate of EEV in the thermodynamic limit (Definition \ref{def}). The weak ETH (Lemma \ref{weaketh}) implies that $R_{f(x)=0}(N)=O(1/\sqrt N)$. If an optimal $\hat f$ exists, then we obtain an upper bound $R_{\hat f}(N)=O(1/\sqrt N)$. Although $R_{\hat f}(N)$ can be identically zero for certain local operators (Example \ref{exa}), we have proved that for almost every local operator, the lower bound $R_f(N)\ge r_f(N)=\Omega(1/N)$ holds for any smooth function $f$ including the optimal $\hat f$ (Theorem \ref{cor}). These results apply to all TI systems in any spatial dimension, regardless of the integrability or chaoticity of the model. In systems satisfying the (strong) ETH (Assumption \ref{ass}), we have proved that the aforementioned lower bound is tight (Theorem \ref{thm:eth}).

An open question is whether the gap between our lower and upper bounds on $R_{\hat f}(N)$ can be reduced or even closed without assuming the ETH (\ref{eq:ass}). To this end, it would be instructive to study $R_{\hat f}(N)$ in (integrable) free-fermion systems, which can be diagonalized analytically and efficiently simulated numerically. We conjecture that in any TI system, $R_{\hat f}(N)=\Theta(1/N)$ for almost every local operator.

\section*{Acknowledgments}

We would like to thank Fernando G.S.L. Brand\~ao, Xie Chen, and Yong-Liang Zhang for collaboration on a related project \cite{HBZ19}. This work was supported by NSF grant PHY-1818914 and a Samsung Advanced Institute of Technology Global Research Partnership.

\appendix

\section{Proof of Lemma \ref{thm}} \label{app}

\begin{lemma} [moments \cite{HBZ19}] \label{l:moment}
For any integer $m\ge0$,
\begin{equation} \label{moment}
\frac1d\sum_jE_j^{2m}=\frac1d\tr(H^{2m})=\Theta(N^m).
\end{equation}
\end{lemma}

\begin{proof}
Expanding $H$ in the generalized Pauli basis, we count the number of terms that do not vanish upon taking the trace in the expansion of $H^{2m}$. There are $\Theta(N^m)$ such terms, the trace of each of which is $\Theta(d)$. Therefore, we obtain Eq. (\ref{moment}).
\end{proof}

This lemma implies that
\begin{multline} \label{eq:trunc}
\frac1d\sum_j|E_j|^m\big|\langle j|A|j\rangle-f(E_j/N)\big|\le\sqrt{\frac1d\sum_jE_j^{2m}\times\frac1d\sum_j\big|\langle j|A|j\rangle-f(E_j/N)\big|^2}\\
=O\big(N^{m/2}r_f(N)\big).
\end{multline}

Almost all eigenstates have vanishing energy density:
\begin{lemma} [concentration of eigenvalues \cite{Ans16}] \label{Mar}
For any $\epsilon>0$,
\begin{equation}
\big|\{j:|E_j|\ge N\epsilon\}\big|/d=e^{-\Omega(N\epsilon^2)}.
\end{equation}
\end{lemma}

This lemma allows us to upper bound the total contribution of all eigenstates away from the middle of the spectrum. Let $C=O(1)$ be a sufficiently large constant such that
\begin{equation} \label{tail}
\frac1d\sum_{j:|E_j|\ge\Lambda}|E_j|^m\le q,\quad\Lambda:=C\sqrt{N\log N},\quad q:=1/\poly(N)
\end{equation}
for $m=0,1,2$, where $\poly(N)$ denotes a polynomial of sufficiently high degree in $N$.

Lemmas \ref{l:moment}, \ref{Mar}, and inequality (\ref{tail}) are related to the fact that $E_j$'s approach a normal distribution in the thermodynamic limit $N\to+\infty$ \cite{KLW15, BC15}. Indeed, $|E_j|=\Theta(\sqrt N)$ for almost all $j$.

For notational simplicity, let $x\overset{\delta}=y$ denote $|x-y|\le\delta$.

\begin{lemma}
For a smooth function $f$, if there exists a strictly increasing infinite sequence $\{N_i\}$ of positive integers such that $r_f(N)=o(1)$ for $N\in\{N_i\}$, then
\begin{equation} \label{zero}
f(0)=0.
\end{equation}
\end{lemma}

\begin{proof}
\begin{align} \label{zeroc}
&0=\frac1d\tr A=\frac1d\sum_j\langle j|A|j\rangle\overset{O(q)}=\frac1d\sum_{j:|E_j|<\Lambda}\langle j|A|j\rangle\overset{r_f(N)}=\frac1d\sum_{j:|E_j|<\Lambda}f(E_j/N)\nonumber\\
&\overset{O(\Lambda^3/N^3)}=\frac1d\sum_{j:|E_j|<\Lambda}f(0)+\frac{f'(0)E_j}N+\frac{f''(0)E_j^2}{2N^2}\overset{O(q)}=\frac1d\sum_jf(0)+\frac{f'(0)E_j}N+\frac{f''(0)E_j^2}{2N^2}\nonumber\\
&=f(0)+\frac{f''(0)\tr(H^2)}{2N^2d}=f(0)+\frac{f''(0)\tr(Hh)}{2Nd},
\end{align}
where we used inequalities (\ref{tail}), (\ref{eq:trunc}), and the Taylor expansion
\begin{equation} \label{taylor}
f(E_j/N)=f(0)+f'(0)E_j/N+f''(0)E_j^2/(2N^2)+O(E_j^3/N^3)
\end{equation}
in the steps marked with ``$O(q)$,'' ``$r_f(N)$,'' and ``$O(\Lambda^3/N^3)$,'' respectively. Since $r_f(N)=o(1)$ for $N\in\{N_i\}$, Eq. (\ref{zero}) follows by letting $N=N_i$ with $i\to+\infty$.
\end{proof}

\begin{lemma}
For a smooth function $f$, if there exists a strictly increasing infinite sequence $\{N_i\}$ of positive integers such that $r_f(N)=o(1/\sqrt N)$ for $N\in\{N_i\}$, then
\begin{equation} \label{deriv}
f'(0)=\tr(HA)/\tr(Hh).
\end{equation}
Note that the right-hand side of this equation is an $N$-independent constant (Lemma \ref{l:1}).
\end{lemma}

\begin{proof}
\begin{multline} \label{derivc}
\frac1d\tr(HA)=\frac{1}{d}\sum_jE_j\langle j|A|j\rangle\overset{O(q)}=\frac1d\sum_{j:|E_j|<\Lambda}E_j\langle j|A|j\rangle\overset{O(\sqrt Nr_f(N))}=\frac1d\sum_{j:|E_j|<\Lambda}E_jf(E_j/N)\\
\overset{O(\Lambda^3/N^2)}=\frac1d\sum_{j:|E_j|<\Lambda}\frac{f'(0)E_j^2}N\overset{O(q)}=\frac1d\sum_j\frac{f'(0)E_j^2}N=\frac{f'(0)\tr(H^2)}{Nd}=\frac{f'(0)\tr(Hh)}d,
\end{multline}
where we used inequalities (\ref{tail}), (\ref{eq:trunc}), and the Taylor expansion (\ref{taylor}) in the steps marked with ``$O(q)$,'' ``$O(\sqrt Nr_f(N))$,'' and ``$O(\Lambda^3/N^2)$,'' respectively. Since $\sqrt Nr_f(N)=o(1)$ for $N\in\{N_i\}$, Eq. (\ref{deriv}) follows by letting $N=N_i$ with $i\to+\infty$.
\end{proof}

We are ready to prove Lemma \ref{thm}. Recalling Eq. (\ref{zeroc}) and using Lemma \ref{l:1},
\begin{equation}
r_f(N)+O(q)+O(\Lambda^3/N^3)\ge|f''(0)|\tr(Hh)/(2Nd)\implies f''(0)=0
\end{equation}
if $r_f(N)=o(1/N)$ for $N\in\{N_i\}$. Then,
\begin{multline} \label{2derivc}
\frac1d\tr(H^2A)=\frac1d\sum_jE_j^2\langle j|A|j\rangle\overset{O(q)}=\frac1d\sum_{j:|E_j|<\Lambda}E_j^2\langle j|A|j\rangle\overset{O(Nr_f(N))}=\frac1d\sum_{j:|E_j|<\Lambda}E_j^2f(E_j/N)\\
\overset{O(\Lambda^5/N^3)}=\frac1d\sum_{j:|E_j|<\Lambda}\frac{f'(0)E_j^3}N\overset{O(q)}=\frac1d\sum_j\frac{f'(0)E_j^3}N=\frac{f'(0)\tr(H^3)}{Nd}=\frac{\tr(HA)\tr(H^2h)}{d\tr(Hh)},
\end{multline}
where we used inequalities (\ref{tail}), (\ref{eq:trunc}), and Eq. (\ref{taylor}) in the steps marked with ``$O(q)$'', ``$O(Nr_f(N))$'', and ``$O(\Lambda^5/N^3)$,'' respectively. Since $Nr_f(N)=o(1)$ for $N\in\{N_i\}$, Eq. (\ref{wrong}) follows by letting $N=N_i$ with $i\to+\infty$.

\section{Proof of Theorem \ref{thm:eth}} \label{app:B}

\begin{lemma} \label{l:8}
For a traceless local operator $A$ with $\|A\|=1$, Assumption \ref{ass} implies that
\begin{gather}
g_N(0)=O(1/N),\label{eq:b1}\\
g_N'(0)=\tr(HA)/\tr(Hh)+O(1/N).\label{eq:b2}
\end{gather}
Note that the first term on the right-hand side of Eq. (\ref{eq:b2}) is an $N$-independent constant (Lemma \ref{l:1}).
\end{lemma}

\begin{proof} [Proof of Eq. (\ref{eq:b1})]
We perform a calculation similar to Eq. (\ref{zeroc}):
\begin{multline} \label{eq:27}
0=\frac1d\tr A=\frac{1}{d}\sum_j\langle j|A|j\rangle\overset{O(q)}=\frac1d\sum_{j:|E_j|<\Lambda}\langle j|A|j\rangle\overset{1/\poly(N)}=\frac1d\sum_{j:|E_j|<\Lambda}g_N(E_j/N)\\
\approx\frac1d\sum_{j:|E_j|<\Lambda}g_N(0)+\frac{g_N'(0)E_j}N\overset{O(q)}=\frac1d\sum_jg_N(0)+\frac{g_N'(0)E_j}N=g_N(0),
\end{multline}
where we used inequality (\ref{tail}), the ETH (\ref{eq:ass}), and the Taylor expansion
\begin{equation} \label{eq:tayloreth}
g_N(E_j/N)=g_N(0)+g_N'(0)E_j/N+g_N''(0)E_j^2/(2N^2)+O(E_j^3/N^3)
\end{equation}
in the steps marked with ``$O(q)$,'' ``$1/\poly(N)$,'' and ``$\approx$,'' respectively. The approximation error in the ``$\approx$'' step is
\begin{equation} \label{eq:29}
\frac1d\sum_{j:|E_j|<\Lambda}O(E_j^2/N^2)\le\frac{1}{d}\sum_{j}O(E_j^2/N^2)=\frac{O(\tr(H^2))}{N^2d}=\frac{O(\tr(Hh))}{Nd}=O(1/N).
\end{equation}
We obtain Eq. (\ref{eq:b1}) by combining Eq. (\ref{eq:27}) and inequality (\ref{eq:29}).
\end{proof}

\begin{proof} [Proof of Eq. (\ref{eq:b2})]
We perform a calculation similar to Eq. (\ref{derivc}):
\begin{align} \label{eq:30}
&\frac1d\tr(HA)=\frac1d\sum_jE_j\langle j|A|j\rangle\overset{O(q)}=\frac1d\sum_{j:|E_j|<\Lambda}E_j\langle j|A|j\rangle\overset{1/\poly(N)}=\frac1d\sum_{j:|E_j|<\Lambda}E_jg_N(E_j/N)\nonumber\\
&\approx\frac1d\sum_{j:|E_j|<\Lambda}E_jg_N(0)+\frac{g_N'(0)E_j^2}N+\frac{g_N''(0)E_j^3}{2N^2}\overset{O(q)}=\frac1d\sum_jE_jg_N(0)+\frac{g_N'(0)E_j^2}N+\frac{g_N''(0)E_j^3}{2N^2}\nonumber\\
&=\frac{g_N'(0)\tr(H^2)}{Nd}+\frac{g_N''(0)\tr(H^3)}{2N^2d}=\frac{g_N'(0)\tr(Hh)}d+\frac{g_N''(0)\tr(H^2h)}{2Nd}\nonumber\\
&=g_N'(0)\tr(Hh)/d+O(1/N),
\end{align}
where we used inequalities (\ref{tail}), (\ref{eq:ass}), and the Taylor expansion (\ref{eq:tayloreth}) in the steps marked with ``$O(q)$,'' ``$1/\poly(N)$,'' and ``$\approx$,'' respectively. The approximation error in the ``$\approx$'' step is
\begin{equation} \label{eq:31}
\frac1d\sum_{j:|E_j|<\Lambda}O(E_j^4/N^3)\le\frac1d\sum_jO(E_j^4/N^3)=O(1/N),
\end{equation}
where we used Eq. (\ref{moment}) with $m=2$. We obtain Eq. (\ref{eq:b2}) by combining Eq. (\ref{eq:30}) and inequality (\ref{eq:31}).
\end{proof}

We are ready to prove Theorem \ref{thm:eth}:
\begin{align}
&r_f^2(N)=\frac1d\sum_j\big|\langle j|A|j\rangle-f(E_j/N)\big|^2\overset{O(q)}=\frac1d\sum_{j:|E_j|<\Lambda}\big|\langle j|A|j\rangle-f(E_j/N)\big|^2\nonumber\\
&\overset{1/\poly(N)}=\frac{1}{d}\sum_{j:|E_j|<\Lambda}|g_N(E_j/N)-f(E_j/N)|^2\approx\frac1d\sum_{j:|E_j|<\Lambda}|g_N(0)+(g'_N(0)-f'(0))E_j/N|^2\nonumber\\
&\le\frac1d\sum_j\big(|g_N(0)|+|g'_N(0)-f'(0)|\cdot |E_j|/N\big)^2=O(1/N^2),
\end{align}
where we used inequalities (\ref{tail}), (\ref{eq:ass}), Eq. (\ref{eq:tayloreth}), and Lemma \ref{l:8} in the steps marked with ``$O(q)$,'' ``$1/\poly(N)$,'' ``$\approx$,'' and the last step, respectively. The approximation error in the ``$\approx$'' step is upper bounded by
\begin{multline}
O(1/d)\sum_{j:|E_j|<\Lambda}|g_N(0)|E_j^2/N^2+|g'_N(0)-f'(0)|\cdot|E_j|^3/N^3+E_j^4/N^4\\
=O(1/N^2+\Lambda^3/N^4+1/N^2)=O(1/N^2),
\end{multline}
where we used Lemmas \ref{l:moment}, \ref{l:8}.

Finally, it is easy to see that $r_f(N)=O(1/N)$ implies that $R_f(N)=O(1/N)$.

\begin{remark}
In a similar way, $R_{f(x)=0}(N)$ can be calculated using Lemma \ref{l:8}; see Lemma 2 in Ref. \cite{Hua22AP}.
\end{remark}

\printbibliography

\end{document}